\documentclass[twocolumn]{autart}    

\usepackage{graphicx}          
\usepackage{mathtools}  
\usepackage{amssymb} 
\usepackage{mathrsfs}
\usepackage{natbib}
\usepackage[bookmarks=false]{hyperref}

\begin{document}

\begin{frontmatter}

\title{Distributionally Robust Stochastic MPC under Disturbance-Affine Feedback Policies}

\author[Aachen]{Xu Chen\thanksref{aut}}\ead{x.chen@irt.rwth-aachen.de}, 
\author[Aachen]{Lorenz Doerschel}\ead{l.doerschel@irt.rwth-aachen.de}
\address[Aachen]{Institute of Automatic Control, RWTH Aachen University, Campus-Boulevard 30, 52074 Aachen}     

\thanks[aut]{Corresponding author.}

\begin{abstract}                          
This study addresses the stochastic Model Predictive Control (MPC) problem for linear time-invariant systems subjected to unknown disturbance distributions. By leveraging the most recent disturbance data, we construct a set of distributions with similar statistical properties contained within a Wasserstein ball, thereby accounting for the worst-case impacts on constraint satisfaction. Numerous MPC strategies, particularly tube-based approaches, have been extensively studied under the Wasserstein ambiguity set, but these methods often introduce conservatism and can limit control performance. Unlike tube-based approaches, we adopt a disturbance-affine control strategy, which introduces additional control degrees of freedom. We begin by developing the Disturbance-Affine Distributionally Robust (DA-DR) MPC framework, subsequently reformulating the control problem into a tractable quadratic programming formulation. Furthermore, we establish the recursive feasibility and stability of the proposed MPC scheme. Finally, we present comprehensive theoretical analysis and simulation results, demonstrating the superiority of the DA-DR MPC over tube-based MPC in initial feasible sets, average performance, and state variance control. 
\end{abstract}

\end{frontmatter}

\section{Introduction}
In many real-world applications of MPC, the actual system is subject to external disturbances. If these disturbances are not accounted for in future predictions, the control actions adopted may become aggressive, potentially severely damaging the performance and safety of the closed-loop system. In such cases, Robust MPC effectively manages the propagation of uncertainties in predictions by considering the worst-case impacts of disturbances within a predefined compact set \cite{mayne2005robust}. However, in practice, disturbances often follow specific probability distributions, and worst-case scenarios typically occur with very low probabilities. Within the Robust MPC framework, control actions tend to be excessively conservative. Consequently, Stochastic MPC more naturally incorporates the distribution of disturbances by employing probabilistic descriptions to define chance constraints, thereby allowing constraint violations with a small probability \cite{li2002probabilistically,farina2016stochastic}. This approach effectively achieves a balance between safety and performance.\par
A significant challenge in Stochastic MPC lies in the formulation of chance constraints in relation to disturbance distributions. When the disturbance distribution is known, or its higher-order statistics are available, the boundaries of these constraints can be approximated through constraint tightening \cite{CANNON2009747,5599849}. However, such approximations are often conservative and rely on the assumption of time-invariant disturbance distributions. Furthermore, some approaches utilize samples of the disturbance data to perform scenario-based optimization \cite{6213078}. In these methods, the probabilistic guarantees for constraint satisfaction depend on the number of selected scenarios, and the theoretical bounds may require sample sizes that exceed practical requirements \cite{1632303,CALAFIORE20131861}. Consequently, the resulting high computational costs can render sample-based MPC methods difficult to implement in real-world applications.\par
By combining the aforementioned approaches, Distributionally Robust Optimization (DRO) serves as an excellent alternative solution framework, as it does not rely on the precise acquisition of the actual distribution. Instead, DRO considers the worst-case within an ambiguity set that encompasses a range of distributions exhibiting characteristics similar to the observed data \cite{rahimian2019distributionally}. Typically, there are two methods for constructing this ambiguity set: moment-based and distance-based methods. Moment-based methods can generally be reformulated into tractable expressions \cite{delage2010distributionally,zymler2013distributionally}. In the context of Stochastic MPC, moment-based methods have been extensively employed \cite{coppens2021data,lu2020soft}. However, such methods often assume that the distribution moments are accurately known, and the resulting solutions may be overly conservative \cite{wang2016likelihood}. Distance-based methods involve considering a family of distributions that are close to the nominal distribution under a specified distance metric. Among these, the Wasserstein distance is particularly prevalent due to its ability to ensure robust out-of-sample performance \cite{mohajerin2018data,zhao2018data}. Numerous studies have effectively integrated the Wasserstein distance with Stochastic MPC to handle soft constraints \cite{aolaritei2023wasserstein,mark2020stochastic,fochesato2022data,mark2021data,aolaritei2023capture}. However, most relevant works exclusively address tube-based strategies for linear systems with additive disturbances, wherein the system dynamics are partitioned into nominal and error dynamics. For the error dynamics, a predetermined feedback gain is applied to control the propagation of errors \cite{langson2004robust}. While this decomposition facilitates problem tractability, it also introduces a degree of conservatism, potentially resulting in suboptimal controller performance under constraint satisfaction. To overcome this limitation, we incorporate a disturbance-affine parameterization strategy \cite{4738806,goulart2006optimization}, thereby establishing our Disturbance-Affine Distributionally Robust MPC (DA-DR MPC). \par
In this work, we develop a tractable reformulation for the DA-DR MPC problem based on the Wasserstein metric. We introduce Conditional Value-at-Risk (CVaR) as an approximation for chance constraints and modify the disturbance-affine feedback strategy accordingly. Concurrently, we address the dual problem and transform it into general linear inequality constraints. This approach ultimately results in a tractable quadratic programming formulation. Additionally, we establish the recursive feasibility of the DA-DR MPC and prove its stability. Furthermore, we conduct a comparative analysis between DA-DR MPC and tube-based methods both theoretically and through simulations. The results show that DA-DR MPC outperforms tube-based MPC in terms of the initial state feasible set, state variance performance, and average performance. 

\section{Preliminaries}
\subsection{Notation}
Let $\mathbb{Z}_{[i,j]} $ denote the set of integers ranging  from $i$ to $j$. The symbol $\mathbb{S}_+$ ($\mathbb{S}_{++}$) denotes the set of positive semidefinite (positive definite) matrices. The notation $\mathrm {int}   (\mathcal{A} )$ represents the interior of the set $\mathcal{A}$. Furthermore, $\mathrm{gph}\ f$ refers to the graph of the function $f$ \cite{rockafellar2009variational}, the symbol $\otimes $ represents the Kronecker product, and the symbol $\oplus$ denotes the Minkowski addition. \par
\subsection{Wasserstein Distance and Ambiguity Set}
In this work, we restrict our focus to the space $\mathbb{R}^n $ and denote by $\mathscr{B}(\mathbb{R}^n)$ the Borel $\sigma$-algebra on  $\mathbb{R}^n $, which allows us to define measures on this measurable space. Although the constructed measure space is generally incomplete, attention is primarily given to sets with positive measures in most applications. Additionally, let $\mathcal{P}(\mathbb{R}^n )$ represent the set of Borel probability measures defined on $\mathbb{R}^n$. In certain cases, we aim to represent the set of probability measures defined on a compact set $\mathcal{K} \in \mathbb{R}^n $  as
\begin{equation}\label{probability_m}
\mathcal{P}(\mathcal{K} ) = \left\{ \mu \in \mathcal{P}(\mathbb{R}^n) : \mu(\mathcal{K}) = 1 \text{ and } \operatorname{supp}(\mu) = \mathcal{K}\right\},
\end{equation}
where the set comprises all Borel probability measures supported on the compact set $\mathcal{K} $.\par
Let $\mathcal{P}_p(\mathcal{K} ) \subseteq \mathcal{P}(\mathcal{K} )$ denote the set of probability measures that possess  a finite $p$-moment where $p \in [1,\infty )$. The $p$-Wasserstein distance with respect to  $\mathcal{P}_p(\mathcal{K} ) $ is then defined as 
\begin{equation}\label{wasserstein_distance}
W_p(\mu , \nu ) := \inf_{\gamma \in \mathcal{P}_p(\mathcal{K}\times\mathcal{K}) }  \left( \int_{\mathcal{K} ^2} d(\xi_1 - \xi_2)^p \gamma(d\xi_1, d\xi_2) \right)^{\frac{1}{p}},
\end{equation}
subject to the constraints that the marginals of $\gamma$ are $\mu$ and $\nu $, respectively. Here, the metric $d$ on the Euclidean space is typically the Euclidean distance, although other distance metrics can be specified as needed. To ensure that the infimum in \eqref{wasserstein_distance} is finite, the metric $d$ is required to be proper, level-bounded, and lower semicontinuous. The first two conditions are generally satisfied by the inherent properties of the distance. Besides, the Wasserstein distance quantifies the minimal effort required to transport a unit mass from the distribution $\mu$ to the distribution $\nu$, which is closely related to the optimal transport problem. \par
In practice, it is often challenging to accurately estimate the actual distribution, especially when the distribution is time-variant. To approximate the disturbance distribution, one approach is to collect recent data and formalize a nominal distribution, such as the empirical distribution defined by $\hat{\mathbb{P}}_N := \frac{1}{N} \sum_{i=1}^{N} \delta_{\hat{\xi}_i}$, based on the available data samples $\{\hat{\xi}_i\}$, where $\delta$ denotes the Dirac delta function. However, the actual distribution may deviate from this nominal distribution. The Wasserstein distance provides a natural way to measure the discrepancy between the two distributions, irrespective of whether they are discrete or continuous. To capture the potential disturbance distribution, we construct a Wasserstein ambiguity set as follows:
\begin{equation}\label{wasserstein_set}
\mathcal{B}_\epsilon(\hat{\mathbb{P}}_N) := \left\{ \nu \in \mathcal{P}_p(\mathcal{K}) : W_p(\nu, \hat{\mathbb{P}}_N) \leq \epsilon \right\},
\end{equation}
with the nominal distribution $\hat{\mathbb{P}}_N$ serves as the center of the ambiguity set. The selected ambiguity radius $\epsilon$, in conjunction with the sample size $N$, determines the probability that the true distribution is contained within $\mathcal{B}_\epsilon$.\par
We now introduce the fundamental duality theory from \cite{gao2023distributionally} in the context of DRO with respect to the Wasserstein ambiguity set centered around the nominal distribution $\hat{\mathbb{P}}_N$. This framework focuses on the impact of the worst-case within the ambiguity set on the objective function $h$. The primal and dual problems for the fixed $\epsilon$ can be formalized as follows:
\begin{equation}\label{prime_dual}
\begin{aligned}
v_P &:= \sup_{\nu \in \mathcal{P}_p(\mathcal{K})} \left\{  \mathbb{E}_\nu\left [ h(\xi  ) \right ]  : \nu \in \mathcal{B}_\epsilon(\hat{\mathbb{P}}_N)  \right\},\\
v_D &:= \inf_{\lambda \geq 0} \left[ \lambda \epsilon^p+ \frac{1}{N} \sum_{i=1}^{N} \sup_{\xi \in \mathcal{K} } \left[ h(\xi) - \lambda d(\xi,\hat{\xi}_i) ^p\right] \right],
\end{aligned}
\end{equation}
where the decision variables in $h$ are omitted for brevity. From the primal infinite-dimensional optimization problem, we derive a finite-dimensional min-max problem. The following theorem in \cite{gao2023distributionally} ensures that strong duality holds under the condition when $\mathcal{K}$ is a compact set.
\begin{thm} \label{zero_dual}
Strong duality in \eqref{prime_dual} holds, i.e., $v_P = v_D$, provided that the growth property $\kappa$ of the function $h$ with the exponent $p$ in \eqref{prime_dual} is bounded above. Specifically, the growth property is defined as 
\begin{equation}\label{growth properties}
\kappa := \limsup_{\xi \in \mathcal{K}:  d(\xi, \xi_0) \to \infty} \frac{h(\xi) - h(\xi_0)}{d(\xi, \xi_0)^p} < \infty,
\end{equation}
for some $\xi_0 \in \mathcal{K}$, where $h$ is upper semicontinuous in $\xi$ and $\epsilon > 0$. Notably, when $\mathcal{K}$ is bounded, $\kappa$ is conventionally  set to $0$.    
\end{thm}
\subsection{Chance Constraint and Its Approximation}
In stochastic MPCs, problems are often subject to chance constraints of the form
\begin{equation} \label{chance-constraint}
\mathbb{P}_\nu(g(x, \xi) \leq 0) \geq 1 - \alpha,
\end{equation}
where $g $ represents a constraint function dependent on the decision variable $x$ and the disturbance $\xi$, and $\alpha$ denotes the allowable violation probability. While regular chance constraint problems are considered to be NP-hard \cite{qiu2014covering}, addressing distributionally robust frameworks remains even more challenging. They are also known to be NP-hard because treating the disturbance set as a singleton reduces the problem to the general case. In this work, we employ CVaR to approximate the chance constraints, thereby facilitating a tractable reformulation \cite{xie2021distributionally}. CVaR is defined as 
\begin{equation} \label{CVaR}
\mathrm {CVaR}_\alpha^{\nu} \left (g(x,\xi )\right ) :=  \inf_{t \in \mathbb{R} }\left [ \alpha^{-1}\mathbb{E}_\nu\left [ \left ( g(x,\xi) +t\right )_+  \right ] -t   \right ],
\end{equation}
where the function $(\cdot ) _+$ defined as $\max\left \{ \cdot ,0 \right \} $. Besides, CVaR quantifies the expected constraint violations occurring beyond the probability threshold $1-\alpha$. 

\section{Problem Formulation}
We consider a discrete-time linear time-invariant system defined by
\begin{equation}\label{system}
x_{k+1} = A x_k + B u_k + D w_k,
\end{equation}
where $x \in \mathbb{R}^{n_x} $ represents the state, $u \in \mathbb{R}^{n_u} $ is the control input, and $w \in \mathbb{R}^{n_w}$ denotes the disturbance. The disturbance $w$ follows an unknown distribution, which may be non-stationary and can change over time. In this work, we assume that its support is compact and forms a convex polytopic set, defined as
\begin{equation}\label{W}
\mathcal{W}:=  \left \{ w \in \mathbb{R}^{n_w}: Hw\le  h\right \} ,
\end{equation}
where $H \in \mathbb{R}^{n_h \times n_w} $, and the partial order is induced by the non-negative orthant. Additionally, the system is subject to affine state constraints
\begin{equation}\label{X}
x \in \mathcal{X}:=   \left \{ x\in \mathbb{R}^{n_x} : Fx\le  f\right \} ,
\end{equation}
where  $F \in \mathbb{R}^{n_f \times n_x} $, and affine input constraints
\begin{equation}\label{U}
u \in \mathcal{U}:=   \left \{ u  \in \mathbb{R}^{n_u}: Gu\le  g\right \} ,
\end{equation}
with  $G \in \mathbb{R}^{n_g \times n_u} $. All matrices $H$, $F$ and $G$ are tall matrices with full rank. \par
In this setting, we address the following discrete-time stochastic optimal control problem:
\begin{equation} \label{PF}
\begin{aligned}
\min_{\mathbf{u}_k} \quad &\mathbb{E}_{ \mathbf{w}_k} \left\{ \sum_{i=0}^{N_h-1}  l\left( \phi_{i|k}(x_k, \mathbf{u}_k, \mathbf{w}_k),u_{i|k}\right) +V_f(\phi_{N_h|k}) \right\}\\
\text{s.t.} \quad &\mathbb{P}_{\mathbf{w}_k}(F \phi_{i|k}(x_k, \mathbf{u}_k, \mathbf{w}_k) \leq f) > 1- \alpha, \ \forall i \in \mathbb{Z}_{[1,N_h]}  \\
&G u_{i|k} \leq g, \  \forall i \in \mathbb{Z}_{[0,N_h-1]} 
\end{aligned}
\end{equation}
where $x_k$ is the system state at time instant $k$, and $N_h$ denotes the prediction horizon. The parameter $\alpha \in \left [ 0,1 \right ) $ indicates the allowable violation probability, with $\alpha = 0$ corresponding to hard constraints. The control input sequence over the prediction horizon is denoted by $\mathbf{u}_k = 
\begin{bmatrix}
u_{0|k}^\mathrm{T}  & \dots & u_{N_h-1|k}^\mathrm{T}
\end{bmatrix}^\mathrm{T}$ while $\mathbf{w}_k = 
\begin{bmatrix}
w_{0|k}^\mathrm{T}  & \dots & w_{N_h-1|k}^\mathrm{T}
\end{bmatrix}^\mathrm{T}$ represents the realization of the disturbance sequence within the prediction, assumed to be i.i.d. with zero means in this work. The constant term can be eliminated by defining a shifted variable. The mapping $\phi$ describes
the state trajectory under the system dynamic constraints in \eqref{system} with the initial state $x_k$. Additionally, $l$ and $V_f$ denote the stage cost and terminal cost, respectively, both of which will be specified in detail later. It is worth noting that, considering real-world conditions, all the input constraints are treated as hard constraints. 
\section{Parameterized Disturbance-Affine Strategy}
In previous discussions on Wasserstein distance-based MPC, tube-based strategies were predominantly employed, enabling the separation of nominal dynamics and error dynamics. For the open-loop optimization strategy, a fixed feedback gain is typically adopted to ensure that the error system matrix is a Schur matrix, which guarantees the contractive behavior of the disturbance propagation. While tubed-based strategies facilitate the tractable formulation of chance constraints, they often introduce conservatism, resulting in a limited initial state feasible set. In this work, to mitigate such conservatism while maintaining computational scalability, we adopt a disturbance-affine control strategy. At each prediction instant $k+i$, the control policy depends on the realizations of disturbances from $k$ to $k+i-1$, increasing the control degree of freedom. Furthermore, we parameterize the control policy using linear control law with respect to the disturbances, which helps to reduce complexity. Specifically, the input sequence over the prediction horizon consists of a gain matrix applied to disturbances to manage the uncertainty propagation, and a bias term to control the nominal behavior, expressed as
\begin{equation} \label{uk}
    \mathbf{u}_k = \mathbf{K}_k \mathbf{w}_k + \mathbf{c}_k,
\end{equation}
\begin{equation} \label{Kk}
\mathbf{K}_k = 
\begin{bmatrix}
0 & 0 & \cdots & 0 & 0 \\
K_{0,1|k} & 0 & \cdots & 0 & 0 \\
K_{0,2|k} & K_{1,2|k} & \cdots & 0 & 0 \\
\vdots & \vdots & \ddots & \vdots & \vdots \\
K_{0,N_h-1|k} & K_{1,N_h-1|k} & \cdots & K_{N_h-2,N_h-1|k} & 0
\end{bmatrix},
\end{equation}
where $\mathbf{c}_k= 
\begin{bmatrix}
c_{0|k}^\mathrm{T}  & \dots & c_{N_h-1|k}^\mathrm{T}
\end{bmatrix}^\mathrm{T}$ is the feedforward control sequence, and both $\mathbf{K}_k$ and $\mathbf{c}_k$ are decision variables that paramiterize the control strategy. Additionally, this expression is equivalent to a state feedback control of the form $\mathbf{u}_k = \mathbf{\tilde{K}}_k \mathbf{x}_k + \mathbf{\tilde{c}}_k$, which can be interchanged by using Youla parameterization \cite{van2002conic}. However, the state feedback control term is nonlinear on the decision variables, as $\mathbf{x}_k$ is itself a function of these variables, thereby introducing non-convexity into the optimization. Thus, we adopt the strategy in \eqref{uk}, which affinely depends on disturbances. \par
Subsequently, the predicted state sequence can be expressed as $\mathbf{x}_k= 
\begin{bmatrix}
x_{1|k}^\mathrm{T}  & \dots & x_{N_h|k}^\mathrm{T}
\end{bmatrix}^\mathrm{T}$, in terms of the control input sequence  $\mathbf{u}_k$ as follows:
\begin{equation} \label{xk}
    \mathbf{x}_k = L_x x_k + L_u \mathbf{c}_k + L_u \mathbf{K}_k \mathbf{w}_k + L_w \mathbf{w}_k ,
\end{equation}
where the matrices $L_x \in \mathbb{R}^{n_xN_h \times n_x}  $, $L_u \in \mathbb{R}^{n_xN_h \times n_uN_h} $ and $L_w \in \mathbb{R}^{n_xN_h \times n_wN_h} $ are defined in the Appendix.

\section{Chance Constraint Reformulation}
Next, we focus on the reformulation of the chance constraints. Here, we consider the probability of constraint satisfaction over the whole prediction horizon. By substituting \eqref{xk} into the chance constraints in \eqref{PF}, we can express the constraints as
\begin{equation} \label{con_xk}
\mathbb{P}_{\mathbf{w}_k }(\mathbf{F}  L_x x_k + \mathbf{F}  L_u \mathbf{c}_k + \mathbf{F}  L_u \mathbf{K}_k \mathbf{w}_k + \mathbf{F} L_w \mathbf{w}_k   \leq \mathbf{f}) \geq 1- \alpha,
\end{equation}
where $\mathbf{F} \in \mathbb{R}^{n_f N_h \times n_x N_h} $ is defined as $\mathbf{F}:=\mathrm{diag(F,\dots ,F)} $, and $\mathbf{f} \in \mathbb{R}^{n_f N_h}$ is defined as $\mathbf{f}:=\left [ f^\mathrm{T}  \dots f^\mathrm{T}\right ]^\mathrm{T}$. We denote $\mathbf{F}_j$ as the $j$-th row of the matrix $\mathbf{F}$ for $j\in \mathbb{Z}_{\left [1, n_f N_h\right ] }  $ and define a finite index set $\mathbb{I}_F := \mathbb{Z}_{\left [1, n_f N_h\right ] }  $. Since the inequality must hold for all $j \in \mathbb{I}_F$, this is equivalent to requiring that the inequality holds for the maximum value of the function over the index set. Then, the chance constraints can be written as
\begin{equation} \label{chance_con_prime}
\begin{aligned}
 &\mathbb{P}_{\mathbf{w}_k }(\mathrm{max}_{j\in\mathbb{I}_F } \ \mathbf{F}_j  L_x x_k + \mathbf{F}_j  L_u \mathbf{c}_k + \mathbf{F}_j  L_u \mathbf{K}_k \mathbf{w}_k\\
   & + \mathbf{F}_j L_w \mathbf{w}_k  - \mathbf{f}_j \leq 0) \geq 1-\alpha.
\end{aligned}
\end{equation}
From this constraint, we observe that the left-hand side of the inequality is a piecewise affine function in both the disturbance $\mathbf{w}_k$ and the decision variables. At this stage, the decision variables are still difficult to handle in chance constraint reformulation due to the matrix form of $\mathbf{K}_k$. To further simplify the problem, we draw upon the method in \cite{chen2023gasoline} to transform the decision variables into a dense vector form. Consequently, the function on the left-hand side can be expressed as 
\begin{equation} \label{chance_con}
\begin{aligned}
&g( \mathbf{c}_k, \mathbf{v}_k, \mathbf{w}_k)  := \\&\mathrm{max}_{j\in\mathbb{I}_F } \ \mathbf{F}_{xj}  x_k + \mathbf{F}_{uj} \mathbf{c}_k + \mathbf{v}_k^{\mathrm {T}   }\mathbf{V}_j^{\mathrm {T}   }\mathbf{w}_k  + \mathbf{F}_{wj} \mathbf{w}_k  - \mathbf{f}_j,
\end{aligned}
\end{equation}
where $\mathbf{F}_{xj}:= \mathbf{F}_j  L_x$, $\mathbf{F}_{uj}:= \mathbf{F}_j  L_u$, $\mathbf{F}_{wj}:= \mathbf{F}_j  L_w$, and $\mathbf{v}_k^{\mathrm {T}   }\mathbf{V}_j^{\mathrm {T}   } = \mathbf{F}_j  L_u \mathbf{K}_k$. Here, $\mathbf{V}_j \in \mathbb{R}^{n_wN_h\times \frac{1}{2} n_wn_uN_h(N_h-1) } $ and $\mathbf{v}_k \in \mathbb{R}^{ \frac{1}{2} n_wn_uN_h(N_h-1) } $ in \cite{chen2023gasoline}. The vectors $\mathbf{v}_k$ and $\mathbf{c}_k$ are both decision variables representing the feedback term and the feedforward term, respectively. \par
Next, we consider the following distributionally robust chance constraint within a Wasserstein ambiguity set, derived from \eqref{chance_con}, defined as
\begin{equation}  \label{robust_chance_con}
\sup_{\nu \in \mathcal{B }_\epsilon }\mathbb{P}_\nu \left(g(\mathbf{c}_k, \mathbf{v}_k, \mathbf{w}_k)\leq 0\right) \geq 1-\alpha.
\end{equation}
By utilizing dual representation and the theory of zero duality gap in Theorem~\ref{zero_dual}, we can derive an equivalent formulation for the chance constraint. However, despite obtaining a theoretically equivalent expression, the resulting problem remains challenging to solve in practice due to its inherent nonconvexity. Therefore, we approximate the original chance constraint in \eqref{robust_chance_con} by employing a probability constraint with CVaR in \eqref{CVaR}, which facilitates a tractable formulation. When the CVaR condition is satisfied, it implies that the original chance constraint also holds. Consequently, the DR chance constraint can be rewritten as
\begin{equation}  \label{CVaR_chance_con_full}
\sup_{\nu \in \mathcal{B }_\epsilon } \inf_{t \in \mathbb{R} }\left [ \alpha^{-1}\mathbb{E}_\nu\left [ \left ( g(\mathbf{c}_k, \mathbf{v}_k, \mathbf{w}_k)+t \right )_+  \right ] -t   \right ] \leq 0.
\end{equation}
Next, we treat the function pointwise with respect to the decision variables. Without loss of generality, we first omit the decision variables and define the function 
\begin{equation}  \label{function_h}
\begin{aligned}
  h(t,\mathbf{w}_k )&:= \left (\tilde{g}(\mathbf{w}_k)+t \right )_+  -\alpha t \\
  & = \max \left \{  \tilde{g}(\mathbf{w}_k)+(1-\alpha)t, -\alpha t\right \} ,
\end{aligned}
\end{equation} 
where $\tilde{g}(\mathbf{w}_k)$ encapsulates the dependence of the constraint function on the disturbance sequence $\mathbf{w}_k$.\par

In the subsequent proof, we need to exchange the order of min and max in \eqref{CVaR_chance_con_full}. We begin by considering the following theorem from \cite{shapiro2002minimax,hota2019data}.
\begin{thm} \label{min-max}
Consider the min-max stochastic optimization problem
\begin{equation} \label{min-max_problem}
    \min_{x \in S}  \sup_{\mu \in \mathcal{A} } \mathbb{E}_\mu \left[ g(x, \xi ) \right] ,
\end{equation}
where $S\subseteq  \mathbb{R} ^n$ is a closed and convex subset, $\mathcal{A}$ is a non-empty set of probability measures on measurable space $(\Xi, \mathcal{B}(\Xi))$, and $g:\mathbb{R}^n\times  \Xi\to \mathbb{R}$ is a real-valued function.\par
Let $\bar{x}$ be an optimal solution of \eqref{min-max_problem}. Suppose that the following conditions hold: 1. the set $\mathcal{A}$ is tight and closed in the weak topology. 2. There is a convex neighborhood $V$ of $S$ such that for all $x \in V$ and $\mu \in \mathcal{A}$ the function $g(x,\cdot)$ is measurable and integrable, the supremum  $p(x):= \sup_{\mu \in \mathcal{A} } \mathbb{E}_\mu \left[ g(x, \xi ) \right] < \infty$, and for all $\xi \in \Xi$ the function $g(\cdot,\xi)$ is convex on $V$. 3. For every $x$ in the neighborhood of $\bar{x}$, the function $g(x,\cdot)$ is bounded and upper semicontinuous on $\Xi$, and the function $g(\bar{x},\cdot)$ is bounded and continuous. \par
Then, there exists $\bar{\mu}$ such that $(\bar{x},\bar{\mu})$ is a saddle point of \eqref{min-max_problem} and the optimal values of \eqref{min-max_problem} and its dual are equal. 
\end{thm}
Based on Theorem~\ref{min-max}, we derive the following lemma, which serves as a critical step in simplifying and analyzing the problem further.
\begin{lem} \label{lem_maxmin}
Consider a Wasserstein ambiguity set $\mathcal{B }_\epsilon$ with respect to the empirical distribution and the constraint function defined in \eqref{chance_con}. For any $\mathbf{c}_k$ and $\mathbf{v}_k$, the following equality holds:
\begin{equation}  \label{CVaR_chance_con_full_dual}
\begin{aligned}
    &\sup_{\nu \in \mathcal{B }_\epsilon } \inf_{t \in \mathbb{R} }\left [ \alpha^{-1}\mathbb{E}_\nu\left [ \left ( g(\mathbf{c}_k, \mathbf{v}_k, \mathbf{w}_k)+t \right )_+  \right ] -t   \right ] \\
    = & \inf_{t \in \mathbb{R} } \sup_{\nu \in \mathcal{B }_\epsilon }\left [ \alpha^{-1}\mathbb{E}_\nu\left [ \left ( g(\mathbf{c}_k, \mathbf{v}_k, \mathbf{w}_k)+t \right )_+  \right ] -t   \right ] 
\end{aligned}
\end{equation}
\end{lem}
\begin{pf}
To prove Lemma~\ref{lem_maxmin}, we verify that the conditions of Theorem~\ref{min-max} are satisfied. 1. The Wasserstein ambiguity set is tight and closed in the weak topology, as proved in \cite{santambrogio2015optimal} by invoking Prokhorov's theorem. 2. Consider the whole space $\mathbb{R}$, for every $t \in \mathbb{R}$, the function $h(t,\cdot)$ defined in \eqref{function_h} is measurable, and by Theorem~\ref{zero_dual}, zero duality gap holds, and the primal problem has a finite optimal value. Specifically, for any $\nu \in \mathcal{B }_\epsilon$, the function is $\nu$-integrable. Furthermore, for all feasible $\mathbf{w}_k$, the function in \eqref{function_h} is also convex in $\mathbf{w}_k$.  The convexity is preserved under the maximization. 3. For any $t \in \mathbb{R}$, the function $h(t, \cdot)$ is a piecewise affine function, which is bounded on the compact set $\mathcal{K}^{N_h} $ and continuous on $\mathrm{int}(\mathcal{K}^{N_h} ) $. Additionally, we need to verify the existence of the infimum of \eqref{min-max_problem} for $t \in \mathbb{R}$. We define the function $p(t) := \sup_{\nu \in \mathcal{B}_\epsilon } \mathbb{E}_\nu\left [ h(t,\mathbf{w}_k ) \right ]  $. Since the function $p(\cdot)$ attains a finite optimum for every $t \in \mathbb{R}$, the function itself is proper. Moreover, for each 
$\nu \in \mathcal{B}_\epsilon$, the function $h$ is lower semicontinuous (lsc) in $t$ because it is continuous. Consequently, the supremum of $h$ over the ambiguity set is also lsc \cite{rockafellar2009variational}. Lastly, we want to show the level-boundedness of the function $p$. To establish this, we refer to \eqref{function_h} and show that $\liminf_{ |  t |\to \infty  } \frac{p(t)}{|  t |} =\min\left \{ \alpha,1-\alpha \right \} $. This finite limit infimum indicates that the function $p$ is level-coercive, which implies that it is level-bounded. Since $p$ is proper, lsc, and level bounded, we can guarantee the attainment of a minimum with a finite value.  \hfill $\square$
\end{pf}
Next, we specify the settings for the ambiguity set. We consider a finite-dimensional normed space equipped with the $L_1$-norm, $\left ( \mathbb{R}^{n_wN_h},\left \| \cdot \right \|_1   \right ) $. In this context,  the set $\mathcal{K}$ is defined as

    $ \mathcal{W}^{N_h}=\left \{  \mathbf{w}_k \in\mathbb{R}^{n_wN_h}: I_{N_h\times N_h}\otimes H \mathbf{w}_k \leq  I_{N_h\times 1}\otimes h\right\}$   .

Additionally, we consider a set of probability measures with the moment $p=1$. Then, the Wasstein distance is given by
\begin{equation}
W(\mu , \nu ) = \inf_{\gamma \in \mathcal{P}_1(\mathcal{K}\times\mathcal{K}) }  \int_{\mathcal{K} ^2} \left \| \xi_1 - \xi_2\right \|_1  \gamma(d\xi_1, d\xi_2) .
\end{equation}

\begin{prop} \label{CVaR_LP}
Consider the function $g( \mathbf{c}_k, \mathbf{v}_k, \mathbf{w}_k)$ as defined in \eqref{chance_con}. Suppose that the Wasserstein distance is defined over the set $\mathcal{W}^{N_h}$ with the moment $p=1$, and the distance metric is induced by the $L_1$-norm. Then, the CVaR approximation of chance constraints in \eqref{CVaR_chance_con_full} can be expressed as linear inequality constraints of the form
\begin{equation} \label{results}
    \begin{aligned}
        &\lambda \epsilon+ \frac{1}{N} \sum_{i=1}^{N}  s_i  -\alpha t\leq0, \quad \lambda \geq 0, \quad n_{ji}\geq 0,\\
        & \hat{\mathbf{w}}_{ki}^{\mathrm{T} }\mathbf{V}_j\mathbf{v}_k + \hat{\mathbf{w}}_{ki}^{\mathrm{T} }\mathbf{F}_{wj}^{\mathrm{T}}-\hat{\mathbf{w}}_{ki}^{\mathrm{T} }(I_{N_h\times N_h}\otimes H)^{\mathrm{T} }n_{ji},\\
&+(I_{N_h\times 1}\otimes h)^{\mathrm{T} }n_{ji}+\mathbf{F}_{xj}  x_k + \mathbf{F}_{uj} \mathbf{c}_k- \mathbf{f}_j + t \leq s_i,\\
& \mathbf{V}_j\mathbf{v}_k+\mathbf{F}_{wj}^{\mathrm{T}}-(I_{N_h\times N_h}\otimes H)^{\mathrm{T} }n_{ji} -\lambda \leq 0,  
    \end{aligned}
\end{equation}
where $s_i$ and $n_{ji}$ are newly introduced variables. Here, $\hat{\mathbf{w}}_{ki}$ denotes the $i$-th sample of the disturbance realization. All inequalities above hold for every $i \in \mathbb{N}_{[1,N]} $ and $j\in\mathbb{I}_F$. Additionally, the last inequalities are interpreted in the partial order induced by the non-negative orthant. 
\end{prop}
\begin{pf}
We first consider the CVaR constraints in \eqref{CVaR_chance_con_full}. Then,
\begin{equation}
\begin{aligned}
&\sup_{\nu \in \mathcal{B }_\epsilon } \inf_{t \in \mathbb{R} }\left [\alpha^{-1} \mathbb{E}_\nu\left [ \left ( \tilde{g}(\mathbf{w}_k)+t \right )_+  \right ] -t   \right ] \leq 0\\
\underset{(1)  }{\Leftrightarrow}   &\inf_{t \in \mathbb{R} }\sup_{\nu \in \mathcal{B }_\epsilon }\left [ \mathbb{E}_\nu\left [ \left ( \tilde{g}(\mathbf{w}_k)+t \right )_+  \right ] -\alpha t   \right ] \leq 0\\
\underset{(2)  }{\Leftrightarrow} &\inf_{t \in \mathbb{R} }\inf_{\lambda \geq 0} [ \lambda \epsilon+ \frac{1}{N} \sum_{i=1}^{N} \sup_{\mathbf{w}_k \in \mathcal{K} } [ \left (\tilde{g}(\mathbf{w}_k)+t \right )_+ \\
&- \lambda \left \| \mathbf{w}_k-\hat{\mathbf{w}}_{ki} \right \|_1 ]   ]-\alpha t\leq0\\
\underset{(3)  }{\Leftrightarrow} & \lambda \epsilon+ \frac{1}{N} \sum_{i=1}^{N}  s_i  -\alpha t\leq0, \quad \lambda \geq 0,\\
&\sup_{\mathbf{w}_k \in \mathcal{K} } \left[ \left (\tilde{g}(\mathbf{w}_k)+t \right )_+ - \lambda \left \| \mathbf{w}_k-\hat{\mathbf{w}}_{ki} \right \|_1 \right] \leq s_i, \\
& \quad\quad\quad\quad\quad\quad\quad\quad\quad
\quad\quad\quad\quad\quad\forall i\in \mathbb{Z}_{[1,N]} 
\end{aligned}
\end{equation}
The equivalence $(1)$, pertaining to the interchange of the maximum and minimum operations, is established by Lemma~\ref{lem_maxmin}. The equivalence $(2)$ follows from the dual formulation with a zero duality gap, as presented in Theorem~\ref{zero_dual}. The final equivalence $(3)$ arises from the introduction of the slack variable $s_i$. Subsequently, we focus on the last inequality condition, which involves two nested maximization problems—one over an index set and the other over the compact set $\mathcal{K}$.  It can be readily verified that this inequality is equivalent to 
\begin{equation} \label{max.chan}
\begin{aligned}
\max \{  &\sup_{\mathbf{w}_k \in \mathcal{K} } \tilde{g}(\mathbf{w}_k)+t - \lambda \left \| \mathbf{w}_k-\hat{\mathbf{w}}_{ki} \right \|_1 , \\ 
&\underset{=0 }{\underbrace{\sup_{\mathbf{w}_k \in \mathcal{K} }- \lambda \left \| \mathbf{w}_k-\hat{\mathbf{w}}_{ki} \right \|_1 }} 
  \}   \leq s_i
\end{aligned}
\end{equation}
By substituting \eqref{chance_con} into the first term in \eqref{max.chan}, we can obtain
\begin{equation} 
\begin{aligned}
&\sup_{\mathbf{w}_k \in \mathcal{K} } \tilde{g}(\mathbf{w}_k)+t - \lambda \left \| \mathbf{w}_k-\hat{\mathbf{w}}_{ki} \right \|_1 \leq s_i\\
\Leftrightarrow & \sup_{\mathbf{w}_k\in \mathcal{K} }\max_{j\in\mathbb{I}_F }  \mathbf{F}_{xj}  x_k + \mathbf{F}_{uj} \mathbf{c}_k + \mathbf{v}_k^{\mathrm {T}   }\mathbf{V}_j^{\mathrm {T}   }\mathbf{w}_k  + \mathbf{F}_{wj} \mathbf{w}_k  - \mathbf{f}_j \\
&\quad \quad\quad\quad\quad\quad+t - \lambda \left \| \mathbf{w}_k-\hat{\mathbf{w}}_{ki} \right \|_1 \leq s_i\\
\Leftrightarrow& \max_{j\in\mathbb{I}_F } \mathbf{F}_{xj}  x_k + \mathbf{F}_{uj} \mathbf{c}_k- \mathbf{f}_j + \sup_{\mathbf{w}_k\in \mathcal{K} }\mathbf{v}_k^{\mathrm {T}   }\mathbf{V}_j^{\mathrm {T}   }\mathbf{w}_k  + \mathbf{F}_{wj} \mathbf{w}_k \\
&\quad \quad\quad\quad\quad\quad - \lambda \left \| \mathbf{w}_k-\hat{\mathbf{w}}_{ki} \right \|_1 \leq s_i -t\\
\underset{(4)  }{\Leftrightarrow} & \inf_{n_{ji}\geq 0,\left \| \mathbf{V}_j\mathbf{v}_k+\mathbf{F}_{wj}^{\mathrm{T}}-(I_{N_h\times N_h}\otimes H)^{\mathrm{T} }n_{ji} \right \|_\infty  \leq\lambda}\hat{\mathbf{w}}_{ki}^{\mathrm{T} }\mathbf{V}_j\mathbf{v}_k \\
+ &\hat{\mathbf{w}}_{ki}^{\mathrm{T} }\mathbf{F}_{wj}^{\mathrm{T}}-\hat{\mathbf{w}}_{ki}^{\mathrm{T} }(I_{N_h\times N_h}\otimes H)^{\mathrm{T} }n_{ji}+(I_{N_h\times 1}\otimes h)^{\mathrm{T} }n_{ji}\\
+&\mathbf{F}_{xj}  x_k + \mathbf{F}_{uj} \mathbf{c}_k- \mathbf{f}_j + t \leq s_i, \forall j\in\mathbb{I}_F,
\end{aligned}
\end{equation}
where $s_i \geq 0$ due to the second term in \eqref{max.chan}. Besides, the equivalence $(4)$ comes from 
\begin{equation}
\begin{aligned}
&\sup_{\mathbf{w}_k\in \mathcal{K} }\mathbf{v}_k^{\mathrm {T}   }\mathbf{V}_j^{\mathrm {T}   }\mathbf{w}_k  + \mathbf{F}_{wj} \mathbf{w}_k- \lambda \left \| \mathbf{w}_k-\hat{\mathbf{w}}_{ki} \right \|_1\\
\underset{(5)  }{=}  &  \sup_{\mathbf{w}_k\in \mathcal{K} }\mathbf{v}_k^{\mathrm {T}   }\mathbf{V}_j^{\mathrm {T}   }\mathbf{w}_k  + \mathbf{F}_{wj} \mathbf{w}_k+ \inf_{\left \| m_{ji} \right \|_\infty  \leq\lambda} m_{ji}^{\mathrm{T} }(\hat{\mathbf{w}}_{ki} - \mathbf{w}_k)\\
\underset{(6)  }{=}  &\inf_{\left \| m_{ji} \right \|_\infty  \leq\lambda}m_{ji}^{\mathrm{T} }\hat{\mathbf{w}}_{ki} + \\
&\inf_{n_{ji} \geq0, (I_{N_h\times N_h}\otimes H)^{\mathrm{T} }n_{ji}=\mathbf{V}_j\mathbf{v}_k+\mathbf{F}_{wj}^{\mathrm{T}}-m_{ji}} (I_{N_h\times 1}\otimes h)^{\mathrm{T} }n_{ji}
\end{aligned}
\end{equation}
The equality $(5)$ stems from the definition of the dual norm. Specifically, for the $L_1$-norm, the corresponding dual norm is the $L_\infty$-norm. Moreover, this equality leverages the relationship between supremum and infimum problems, i.e., $\inf_{x \in X} f(x) = - \sup_{x \in X}  - f(x) $. The equality $(6)$  is derived from the dual formulation of the linear programming problem $\sup_{\mathbf{w}_k\in \mathcal{K} }\mathbf{v}_k^{\mathrm {T}   }\mathbf{V}_j^{\mathrm {T}   }\mathbf{w}_k  + \mathbf{F}_{wj} \mathbf{w}_k-  m_{ji}^{\mathrm{T} }\mathbf{w}_k$, where dual variables $n_{ji}$ are introduced. \hfill $\square$
\end{pf}
According to Proposition~\ref{CVaR_LP}, the chance constraints in \eqref{CVaR_chance_con_full} can be transformed into tractable linear inequalities. In the next section, we consider an objective function of the quadratic form. This ultimately leads to a standard quadratic programming problem that can be solved efficiently.

\section{Properties of DA-DR MPC}
In this section, we aim to demonstrate the properties of DA-DR MPC, such as recursive feasibility and stability. We first establish recursive feasibility for MPC by adding additional constraints. Furthermore, we prove that the long-term average cost is upper bounded.
\subsection{Recursive Feasibility}
In the context of stochastic MPC, ensuring recursive feasibility constitutes a significant challenge. When dealing with probability distributions with unbounded support, such as the Gaussian distribution, achieving recursive feasibility becomes generally unattainable because the worst-case scenarios remain unbounded despite the low probability of such extreme events in light-tailed distributions. This limitation motivates our approach to restrict the support of disturbances to a bounded set. Additionally, within the original problem in \eqref{PF}, considering all disturbance realizations within a bounded support does not inherently guarantee recursive feasibility, as the constraints are satisfied only with a specified probability not for all possible valid realizations. Therefore, it is necessary to modify the original problem by incorporating additional constraints to ensure recursive feasibility. To address this issue, we first introduce $t$-step-ahead constraints.
\begin{defn} \label{n_step_constraints}
At the time instant $k$, $t$-step-ahead constraints require that,  for any realization of the disturbance over the first $t$ time steps, the constraints defined in \eqref{PF} must be satisfied from time $k+t+1$ to time $k+t+N_h$ with the predefined probability $1-\alpha$. 
\end{defn}
 These additional constraints ensure that future states remain feasible for all disturbance realizations within the first $t$ step from the current instant $k$. Furthermore, we introduce a terminal set as detailed in the following assumption. \par
\begin{assum} \label{terminal_set_1}
There exists a terminal set $\mathcal{X}_f$ and a linear feedback gain $K_f$, such that the following conditions are satisfied:
\begin{equation}  \label{terminal_set}
\begin{aligned}
 &(A+BK_f)\mathcal{X}_f \oplus D\mathcal{W} \subseteq \mathcal{X}_f,\\ 
&K_f\mathcal{X}_f \subseteq \mathcal{U}, \quad \mathcal{X}_f \subseteq \mathcal{X}.    
\end{aligned}
\end{equation}
\end{assum}
We aim to determine the maximal robust positively invariant (RPI) set for the terminal set, adhering to the system dynamics described in \eqref{system} and subject to the hard constraints specified in \eqref{X} and \eqref{U}. The construction follows the methodology presented in \cite{kouvaritakis2016model}.\par
Within the framework of DA-DR MPC, we construct the following candidate solution at time instant $k+1$:
\begin{equation} \label{Kk+1}
    \tilde{\mathbf{K}}_{k+1} := 
\begin{bmatrix}
0 & 0 & \cdots & 0 & 0 \\
K_{1,2|k}^{\ast}  & 0 & \cdots & 0 & 0 \\
K_{1,3|k}^{\ast} & K_{2,3|k}^{\ast} & \cdots & 0 & 0 \\
\vdots & \vdots & \ddots & \vdots & \vdots \\
K_{1,N_h-1|k}^{\ast} & K_{2,N_h-1|k}^{\ast} & \cdots & 0 & 0\\
K_fA^{N_h-2}D&K_fA^{N_h-3}D &\cdots & K_fD &0
\end{bmatrix}.
\end{equation}
In the matrix $ \tilde{\mathbf{K}}_{k+1}$, the final row represents the feedback w.r.t. the error at the $N_h-2|k+1$-step prediction, and we define the feedforward term as follows:
\begin{equation} \label{ck+1}
\tilde{\mathbf{c}}_{k+1}:= 
\begin{bmatrix}
c_{1|k}^{\ast }+K_{0,1|k}^{\ast }w_k\\
\vdots \\
c_{N_h-1|k}^{\ast }+K_{0,N_h-1|k}^{\ast }w_k\\
K_f\bar{x}_{N_h-1|k+1} 
\end{bmatrix},
\end{equation}
where the last component is the feedback w.r.t. the nominal state $\bar{x}_{N_h-1|k+1} := A^{N_h} x_k + A^{N_h-1} D w_k + A^{N_h-1} B c_{0|k}^{\ast } + \cdots + B c_{N_h-1|k}^{\ast }$. Furthermore, the candidate solution for 2-step-ahead prediction can be formulated analogously.\par
Next, we address the following consideration: by introducing additional $1$-step-ahead constraints in Definition~\ref{n_step_constraints} to the original problem described in \eqref{PF} with the CVaR reformulation in \eqref{CVaR_chance_con_full}, denoted as $P_0$, we formulate an extended problem referred to as $P_1$. The general definition of $P_t$ is provided below.
\begin{defn} \label{P_n}
At time instant $k$, the original problem described in \eqref{PF} is designated as $P_0$. The problem $P_t$ is constructed by sequentially adding constraints from $1$-step-ahead up to $t$-step-ahead.
\end{defn}
At time step $k$, provided that $P_1$ is initially feasible, the solution of $P_1$ is characterized by the matrices $\mathbf{K}^{\ast}_{k}$ and $\mathbf{c}^{\ast}_{k}$. We then consider the candidate solution as defined in \eqref{Kk+1} and \eqref{ck+1}. At the subsequent time step $k+1$, for any realization of the disturbance, this candidate continues to satisfy all constraints with the given probability. However, it is important to note that this candidate solution may not necessarily be the solution to $P_1$, as there is no guarantee that it satisfies the $1$-step-ahead constraints in $P_1$. Consequently, we cannot sufficiently guarantee that $P_1$ is always feasible. Similarly, for the problem $P_t$, we cannot ensure that the corresponding candidate solution for $t$-step-ahead prediction is the solution to $P_t$. Through induction, we conclude that we need to define a problem $P_\infty$ to sufficiently ensure that $P_\infty$ is recursively feasible. The construction of recursive feasibility is similar to the approach in the related work \cite{kouvaritakis2010explicit}, although the starting points are different. When hard terminal set constraints in Assumption ~\ref{terminal_set} are introduced at prediction step $N_h$, it suffices to consider only the problem $P_{N_h}$. This is because all constraints are robustly satisfied for all $i \geq N_h$. \par
Next, we examine the $t$-step-ahead constraints in greater detail for $t \in \mathbb{Z}_{[1,N_h]} $.  Given that we impose hard constraints on the terminal set at time step $N_h|k$, it is necessary to adjust the matrices $L_x$, $L_u$, and $L_w$ by removing the last row, which corresponds to the state $x_{N_h|k}$. However, to maintain notational simplicity, we retain the original expressions for $L_x$, $L_u$, and $L_w$. For the matrix $\mathbf{F}\in \mathbb{R}^{n_f(N_h-1) \times n_x(N_h-1)}$ as defined in \eqref{con_xk}, our focus is exclusively on the chance constraints from step $t$ onward. Specifically, we consider only the last $N_h-1-t$ blocks of $\mathbf{F}$, denoted as $\tilde{\mathbf{F}}_t \in \mathbb{R}^{n_f(N_h-1-t) \times n_x(N_h-1)} $. The corresponding index set for these constraints is denoted by $\mathbb{I}_{\tilde{\mathbf{F}}} := \mathbb{Z}_{\left [1, n_f(N_h-1-t)\right ] }$. Moreover, it is imperative to account for the impact of the first $t$ disturbance realizations on each chance constraint, which can be expressed as follows:
\begin{equation} \label{fb_con}
\begin{aligned}
    &\max_{\mathbf{w}_{[0,t-1]} \in \mathcal{W}^{t}}\tilde{\mathbf{F}}_{tj} L_x x_k +\tilde{\mathbf{F}}_{tj}  L_u \mathbf{c}_k +  \tilde{\mathbf{F}}_{tj}  L_u \mathbf{K}_k \mathbf{w}_k\\
   & +  \tilde{\mathbf{F}}_{tj} L_w \mathbf{w}_k  -  \tilde{\mathbf{f}}_{tj} \leq 0 \\
\Leftrightarrow \ 
& \tilde{\mathbf{F}}_{xj}x_k+\tilde{\mathbf{F}}_{uj} \mathbf{c}_k + \mathbf{v}_k^{\mathrm {T}   }\tilde{\mathbf{V}}_j^{(2)\mathrm {T}   }\mathbf{w}_{[t,N-1]}+\tilde{\mathbf{F}}_{wj}^{(2)} \mathbf{w}_{[t,N-1]}\\
&-\tilde{\mathbf{f}}_{tj}+\left ( I_{t\times 1}\otimes h \right ) ^{\mathrm{T} }y_{tj}\leq 0, \quad y_{tj} \geq  0,\\
&\left( I_{t\times t}\otimes h \right ) ^{\mathrm{T} }y_{tj}- \tilde{\mathbf{V}}_j^{(1)  }\mathbf{v}_k-\tilde{\mathbf{F}}_{wj}^{(1)\mathrm{T}} \geq 0,\\
\end{aligned}
\end{equation}
The intermediate steps of the proof can be found in the Appendix. Furthermore, we observe that only the first inequality in the reformulated form in \eqref{fb_con} is related to the uncertainty after step $t$. Consequently, this inequality is treated as a chance constraint with respect to uncertainty. The subsequent steps about the chance constraint reformulation follow a similar pattern as the previously described process. Then, we conclude with the recursive feasibility of DA-DR MPC in the following proposition. 
\begin{prop} \label{fb}
Suppose that Assumption~\ref{terminal_set_1} holds and serves as the terminal constraint. If problem $P_N$ in Definition~\ref{P_n} is formulated accordingly, then DA-DR MPC is recursively feasible.
\end{prop}
\begin{pf}
    Based on the construction process outlined above, we can infer that if the MPC has an initial feasible solution for $P_N$, then the candidate solution at each subsequent step is also a feasible solution for the MPC.\hfill $\square$
\end{pf}
\begin{rem}
In comparison to tube-based methods, constructing recursive feasibility within DA-DR MPC presents greater complexity due to the coupling between the decision variables and error dynamics. This approach introduces a substantial number of additional $t$-step-ahead constraints and variables, significantly increasing the potential computational burden. \par
Moreover, within the $t$-step-ahead constraints, as seen in step $(b)$, the worst-case disturbance is dependent on the variable $\mathbf{v}_k$. If we aim to decouple this relationship, such as by considering the vertices of the disturbance set, then the number of constraints grows exponentially with $t$. \par
Additionally, in the $t$-step-ahead constraint, we only require that the soft constraints hold from $k+t+1$ to $k+t+N_h-1$ under any realization of the disturbances over the initial $t$ steps, with the state reaching the terminal set at $k+t+N_h$. No other hard constraints are imposed. For the robust MPC, we need to ensure that constraints are satisfied at each prediction step under the worst case realizations, which implies that the solution of the robust case remains feasible within problem $P_{N_h}$.  
\end{rem}
\subsection{Objective Function and Stability}
Furthermore, to achieve a tractable formulation, we employ a sample average mean cost function, as illustrated below:
\begin{equation}
\begin{aligned}
    & \mathbb{E}_{\hat{\mathbb{P}}_N} \left\{ \sum_{i=0}^{N_h-1}  l\left( \phi_{i|k},u_{i|k}\right) +V_f(\phi_{N_h|k}) \right\}=\\
& \frac{1}{N}  \sum_{l=1}^{N} \left\{ \sum_{i=0}^{N_h-1}  \left ( \left \| \phi_{i|k}^{(l)} \right \|_Q^2+ \left \|  u_{i|k}^{(l)} \right \|_R^2\right )   + \left \| \phi_{N_h|k}^{(l)} \right \|_P^2 \right\},
\end{aligned}
\end{equation}
where $Q \in \mathbb{S}_+^{n_x} $, $R \in \mathbb{S}_{++}^{n_u} $, and $P$ is solution to the Lyapunov equation $(A+BK_f)^{\mathrm{T} }P(A+BK_f)-P=-Q-K_f^{\mathrm{T} }RK_f$. \par
Next, we consider the stability issue of DA-DR MPC. Due to the inherent additive disturbance in the system, the states $x_k$ and control inputs $u_k$ will not asymptotically converge to the origin. To characterize the asymptotic behavior of $x_k$ and $u_k$, we will use a long-term average cost function, as outlined in the following proposition. Let $J^\ast $ represent the optimal value of the cost function, and $\tilde{J} $ denote the value of the cost function obtained by using the candidate solution in \eqref{Kk+1} and \eqref{ck+1}. \par
\begin{prop} \label{stability}
Suppose that Assumption~\ref{terminal_set_1} holds and serves as the terminal constraints, and consider the problem $P_N$ of DA-DR MPC in Definition~\eqref{P_n}. Then, the following holds:
\begin{equation}
     \limsup_{T\to \infty }\frac{1}{T}\sum_{k=0}^{T}\mathbb{E}\left [\left \| x_k \right \|_Q^2+  \left \| u_k \right \|_R^2 \right ] \leq \mathrm{tr} (\Sigma_wP ),
\end{equation}
where $\Sigma_w$ represents the sample covariance $\Sigma_w$ of $w$.
\end{prop}
\begin{pf}
    We follow the ideas for tube-based methods in \cite{hewing2020recursively,kouvaritakis2010explicit}. First, we consider the expected value of the discrepancy between $J^\ast(x_{k+1})$ and $J^\ast(x_{k})$, conditioned on $x_k$, expressed as
    \begin{equation}
        \mathbb{E} \left [ J^\ast(x_{k+1})- J^\ast(x_{k})| x_k\right] 
    \end{equation}
    Since all disturbances are conditional on $x_k$, the term $w_k$ does not represent the actual realization of disturbance at time instant $k$, but rather follows a given distribution, just as disturbance in the prediction. Moreover, we assume that the considered disturbances are i.i.d. with zero mean. This discrepancy is upper-bounded in the form of the following inequality:
    \begin{equation} \label{upper_b}
\begin{aligned}
 \mathbb{E} \left [ J^\ast(x_{k+1})- J^\ast(x_{k})| x_k\right]    
\leq \mathrm {tr}(\Sigma_w P )-\left \| x_{k} \right \|_Q^2-\left \| u_{k} \right \|_R^2,
\end{aligned}
    \end{equation}
where $\Sigma_w$ denotes the sample covariance of $w$. The proof of the inequality in \eqref{upper_b} can be found in the Appendix. Next, we examine the expected value of $J^\ast(x_{k+1})$, conditional on $x_0$. It is equivalent to the expectation as shown below:
\begin{equation} \label{exp_val}
    \begin{aligned}
 \mathbb{E} \left [ J^\ast(x_{k+1})| x_0\right]    
= \mathbb{E} \left [ \mathbb{E} \left [ J^\ast(x_{k+1})| x_k\right] | x_0 \right]  ,
\end{aligned}
\end{equation}
The detailed derivation process is also provided in the Appendix. Next, we sum both sides of \eqref{upper_b} from $k=0$ onwards. Simultaneously, based on \eqref{exp_val}, we take the expectation of both sides of the inequality conditioned on $x_0$. This yields the following result:
\begin{equation} \label{soft_bounded}
     \limsup_{T\to \infty }\frac{1}{T}\sum_{k=0}^{T}\mathbb{E}\left [\left \| x_k \right \|_Q^2+  \left \| u_k \right \|_R^2 |x_0 \right ] \leq \mathrm{tr} (\Sigma_wP ).
\end{equation}
Therefore, the long-term average cost function is upper-bounded.\hfill $\square$
\end{pf} 

\subsection{Performance Comparison of DA-DR MPC and Tube-based MPC}
In this section, we primarily investigate the impact of  DA-DR MPC and tube-based MPC approaches on the initial state feasible set. It is evident that tube-based DR MPC with constant state feedback gain is a special case of DA-DR MPC. First, we define the set-valued mapping $P: \mathbb{R}^{n_v} \Rightarrow \mathbb{R}^{n_x\times n_c\times n_\rho} $ as
\begin{equation} \label{P}
    P(\mathbf{v} ):=\left \{ (x,\mathbf{c},\rho):f_i(x,\mathbf{c},\mathbf{v},\rho )\leq0, i\in \mathbb{I}  \right \}
\end{equation}
for $\mathbf{v} \in  \mathcal{V}$, where $\mathbf{v}$ is defined in such a way that if $\mathbf{v}  \notin  \mathcal{V}$, then $P(\mathbf{v}) = \emptyset $. Here, $n_v = \frac{1}{2} n_wn_uN_h(N_h-1)$ corresponds to the feedback control, and $n_c =n_uN_h$ corresponds to the feedforward control. We summarize all inequalities from \eqref{results}, \eqref{terminal_set}, and the input constraints from \eqref{PF} as inequalities $f_i(x,\mathbf{c},\mathbf{v},\rho )\leq0$ for the corresponding index set $\mathbb{I}$, where $\rho \in \mathbb{R}^{n_\rho}$ is defined as the relevant slack variables. If the tube-based MPC as a special case is feasible, then the corresponding $\mathbf{v}$ lies within the set $\mathcal{V}$. Next, we compare the initial state feasible set between tube-based MPC and DA-DR MPC. First, we consider the following lemma.\par
\begin{lem}
The set $\mathcal{V}$ is compact, and the set-valued mapping $P$ defined in \eqref{P} is outer semicontinuous (osc). 
\end{lem} \label{osc}
\begin{pf}
    We can easily verify that every limit point of any sequence $\{\mathbf{v}_n\}\subset \mathcal{V} $ belongs to $ \mathcal{V}$ due to the non-strict inequality constraints, implying that the set is closed. Now, assume $\mathbf{v}$ is unbounded. Since $\mathbf{K}_k$ in \eqref{Kk} is the reformulation of $\mathbf{v}$ and the set $\mathcal{W}$ is non-empty, it brings an unbounded control input. This contradicts the fact that the input set in \eqref{U} is compact. Next, the graph of the mapping $P$ is given by:
    \begin{equation}
        \mathrm{gph} \ P=\mathbb{R}^{n_x\times n_c\times n_\rho}\times    \mathcal{V}\cap \left ( \cap_{i\in\mathbb{I} }\left \{ (x,\mathbf{c},\rho,\mathbf{v}  ) :f_i\leq0\right \} \right ) 
    \end{equation}
    which forms a closed set. Therefore, the mapping $P$ is osc. \hfill $\square$
\end{pf}
Then, we come to the following result regarding the initial state feasible set.
\begin{thm}
    Suppose that the DA-DR MPC problem is subject to the inequality constraints in  \eqref{results}, \eqref{terminal_set} and the input constraints in \eqref{U}, denoted as $f_i$ in \eqref{soft_bounded} with the index set $\mathbb{I}$. If there exists a tubed-based strategy $\hat{\mathbf{v}} $ such that the solutions for $x$ and $\mathbf{c}$ are feasible, and there exist $\hat{x} $, $\hat{\mathbf{c}} $, and $\hat{\rho} $ such that $\mathrm {int} \ P(\hat{\mathbf{v} })  \ne \emptyset  $, i.e., $f_i(\hat{x},\hat{\mathbf{c} },\hat{\mathbf{v} },\hat{\rho}) < 0$ for every $i \in \mathbb{I} $, then the initial feasible set of $x$ under DA-DR MPC is at least as large as that under the tubed-based strategy. 
\end{thm}
\begin{pf}
First, we aim to prove the continuity of the mapping $P$. For simplicity, we denote $z \in \mathbb{R}^{n_z} $ as $(x, \mathbf{c},\rho)$ and define the function $f(z,\mathbf{v} ):= \max_{i\in\mathbb{I} }f_i(z,\mathbf{v})$, which is also convex and continuous due to the convexity and continuity of $f_i$. Then, for every $\mathbf{v}$, $P(\mathbf{v})$ is convex. Since there exist $\hat{z}$ and $\hat{\mathbf{v}} $ such that $f(\hat{z}, \hat{\mathbf{v}} ) < 0$, then there is an open set $\mathcal{O}$ containing $\hat{\mathbf{v}}$ such that for any $\mathbf{v} \in \mathcal{O}$, $f(\hat{z}, \mathbf{v}) <0$. This implies that $\mathrm{int}  P(\mathbf{v})\ne \emptyset ,\forall \mathbf{v} \in\mathcal{O} $. Then, we consider any $\tilde{\mathbf{v}} \in \mathcal{O}$ and $\tilde{z} \in \mathrm{int} P(\tilde{\mathbf{v}})$. Due to the continuity of $f$, there exists a open set $\mathcal{Q} $ of $(\tilde{z}, \tilde{\mathbf{v}})$ with $\mathcal{Q} \subset \mathrm{int} P(\tilde{\mathbf{v}})\times \mathcal{O} $ such that the set $\mathcal{Q} \subset \mathrm {gph}P  $ with $f<0 $ on $\mathcal{Q}$. This implies that $\forall  \mathbf{v}^k\rightarrow  \tilde{\mathbf{v}}, \exists N\in \mathcal{N}_\infty , z^k\underset{N}{\rightarrow} \tilde{z} $ with $ z^k\in P(\mathbf{v}^k) $, i.e., $\tilde{z}\in \liminf_{\mathbf{v}\rightarrow\tilde{\mathbf{v}}  }P(\mathbf{v})$. Then, $\mathrm {int} P(\tilde{\mathbf{v} })   \subseteq \liminf_{\mathbf{v}\rightarrow\tilde{\mathbf{v}}  }P(\mathbf{v})$. Because the inner limit is always closed and the closure $\mathrm {cl}\left (   \mathrm {int} \ P(\tilde{\mathbf{v} })   \right ) = P(\tilde{\mathbf{v} })$, the mapping $P$ is inner semicontinuous at $\tilde{\mathbf{v} }$ when we take closure of both sides. Based on the result regarding the osc in Lemma~\ref{osc}, we can conclude that the mapping $P $ is continuous at $\tilde{\mathbf{v} }$. Furthermore, the continuity can be extended to any $\mathbf{v}$ with $   \mathrm {int} \ P(\mathbf{v} )   \ne \emptyset $.\par
Next, we consider the projection $T: (x,\mathbf{c},\rho )\rightarrow x$. The composite mapping $T\circ P$ remains a set-valued mapping, and the projection preserves the continuity of the mapping. Thus, we obtain the mapping from $\mathbf{v}$ to the initial state feasible set. For a specific value of $\mathbf{v}$, such as the tube-based strategy with constant state feedback, we can derive the corresponding state feasible set. We now characterize the volume of the initial state feasible set. For any $\mathbf{v} \in \mathcal{V}$, the set under the mapping $T\circ P$ is measurable. We define the following Lebesgue integral $L(\mathcal{A}):=\int_{\mathbb{R}^{n_x} } \chi_{\mathcal{A}}(x)  d\mu(x)$, where $\chi$ is the indicator function. When the set 
$\mathcal{A}$ varies continuously, the integral is also continuously dependent on $\mathcal{A}$. Eventually, we can establish a continuous composite function $L\circ T\circ P$ that measures the volume of the initial feasible set under feedback control $\mathbf{v}$. It is worth mentioning that since the state set, input set, and $\mathcal{V}$ are all compact, the initial feasible set is also compact, which means that the integrals all have finite values. Subsequently, for any compact set encompassing $\hat{\mathbf{v}} $ within the set $\mathcal{V}$, we can always find a maximum value with respect to the mapping $L\circ T\circ P$. This implies that the initial state feasible set resulting from DA-DR MPC is at least as large as that obtained by the tube-based strategy. \hfill $\square$
\end{pf}
\begin{rem}
In cases where the tube-based strategy does not have a feasible solution, we cannot conclusively determine that the initial state feasible set under DA-DR MPC is empty. However, the converse does hold true.  Furthermore, when we consider a horizon $N_h=1$, the DA-DR MPC becomes trivially equivalent to the tube-based strategy.
\end{rem}

\section{Simulative Example and Results}
Next, we consider the application of DA-DR MPC in the context of Gasoline-Controlled Autoignition (GCAI). GCAI is an innovative lean combustion technology where fuel and air are uniformly mixed during the compression phase, and the gas mixture compresses until autoignition occurs. Besides, the combustion takes place simultaneously throughout the combustion chamber without a high-temperature ignition flame while maintaining a high air-fuel ratio, which significantly reduces the formation of NOx and particulates \cite{oakley2001experimental}. However, due to the lack of a direct ignition mechanism and strong cyclic coupling, even a small outlier can lead to irregular combustion cycles, which makes control challenging. Due to the stochastic nature of the combustion process, it is essential to incorporate process uncertainties into our MPC problem formulation to enhance control performance.\par
\subsection{System Description and Evaluation}
In this study, we consider Negative Valve Overlap (NVO) \cite{koopmans2003direct}, the duration of fuel injection, and water injection as our control inputs. The states are Center of Combustion (CA50), Indicated Mean Effective Pressure (IMEP), and Maximum Pressure Gradient (DPMAX). CA50, which is crucial for combustion efficiency and stability, is to be stabilized at a set point, such as $7\ \text{°CA}$ after Top Dead Center (aTDC). IMEP, related to output torque, is subject to reference tracking, while DPMAX is constrained within a specific range to maintain safe combustion conditions.\par
Then, we consider the following linearized model based on the work \cite{nuss2019reduced} for the normalized system:
\begin{equation}
    A=\begin{bmatrix}
 0.15 & 0.84 & -0.57\\
 0.20 & 0.23 & 0.16\\
 -0.01 & -0.15 & 0.35
\end{bmatrix},  
B=\begin{bmatrix}
 -0.14 & 0.40 & -1.27\\
 0.65 & -0.08 & 0.17\\
 1.14 & -0.23 & 0.91
\end{bmatrix},
\end{equation}
where $A$ is system matrix and $B$ is input matrix. Additionally, we assume that the external disturbance primarily affects CA50 and IMEP directly, while having no direct impact on DPMAx. Then, the allocation of disturbances is governed by the corresponding matrix $D$. In this work, we adopt the uncertainty model presented in \cite{chen2023gasoline} as the generator of disturbance. \par
In the following experiments, the radius of the Wasserstein ball is set to $0.008$, and $15$ samples of disturbance trajectories are used. The preset violation tolerance is $\alpha=0.1$, and the prediction horizon $N_h$ is set to $4$ due to the rapid dynamic behavior of the process. The weighting matrix $Q$ is set to $\mathrm {diag}(10,10,1)  $, and the matrix $R$ is set to $\mathrm {diag}(1,1,1)  $. Besides, the permissible range for CA50 spans from $4$ °CA to $11$ °CA, while the feasible range for IMEP lies between $2$ and $4$ bar. Additionally, DPMAX is constrained to a range of $0$ to $5$ bar/°CA.
Furthermore, the control objective for IMEP is to track a time-varying reference, while the objective for CA50 is to maintain it at a set point of $7$ °CA aTDC. \par
Regarding disturbances $w$, another essential task is to determine the boundaries of the disturbance set. This involves characterizing a convex outer approximation of the set using multiple inequalities, thereby forming a polytope. A substantial amount of disturbance data can be obtained from testbed measurements, which captures model mismatches, inherent process stochasticity, and noise. To balance the accuracy of the set description with computational complexity, we apply the PCA techniques outlined in \cite{cheramin2021data} to the collected dataset. This approach allows us to derive a simplified and manageable disturbance set. \par
To evaluate the performance, we conducted 60 simulation runs, each subject to randomly sampled disturbance trajectories. Our objective was to assess the control performance of DA-DR MPC and tube-based MPC under the identical Wasserstein ambiguity set. In the case of tube-based MPC, the feedback gain is derived by solving an infinite-horizon algebraic Riccati equation. The primary performance metrics of interest are the stability of CA50 and the tracking capability of IMEP. For CA50, our goal is to maintain its value within the specified constraint range. Excessively high CA50 values indicate incomplete combustion, while excessively low CA50 values signify early combustion, both of which are undesirable outcomes. Therefore, we examine the probability of CA50 samples violating the upper and lower bounds, as well as the overall variance of CA50 in the closed-loop behaviors under these two controllers. Regarding IMEP, due to its relatively lower uncertainty, we primarily focus on its average performance across simulation samples. We compare the two controllers based on their ability to respond to and track time-varying IMEP references. The controller that exhibits minimal tracking offsets demonstrates superior tracking performance. These metrics provide a evaluation of the robustness and efficiency of DA-DR MPC and tube-based MPC in managing the stochastic nature of the GCAI system.\par

\subsection{Results and Discussion}
First, we focus on the quantification of the disturbance set. Based on extensive analysis of historical data, the marginal distribution of the disturbance is bounded, as shown in the figure below.\par
\begin{figure}[htb]
\centering
\includegraphics[width=0.475\textwidth]{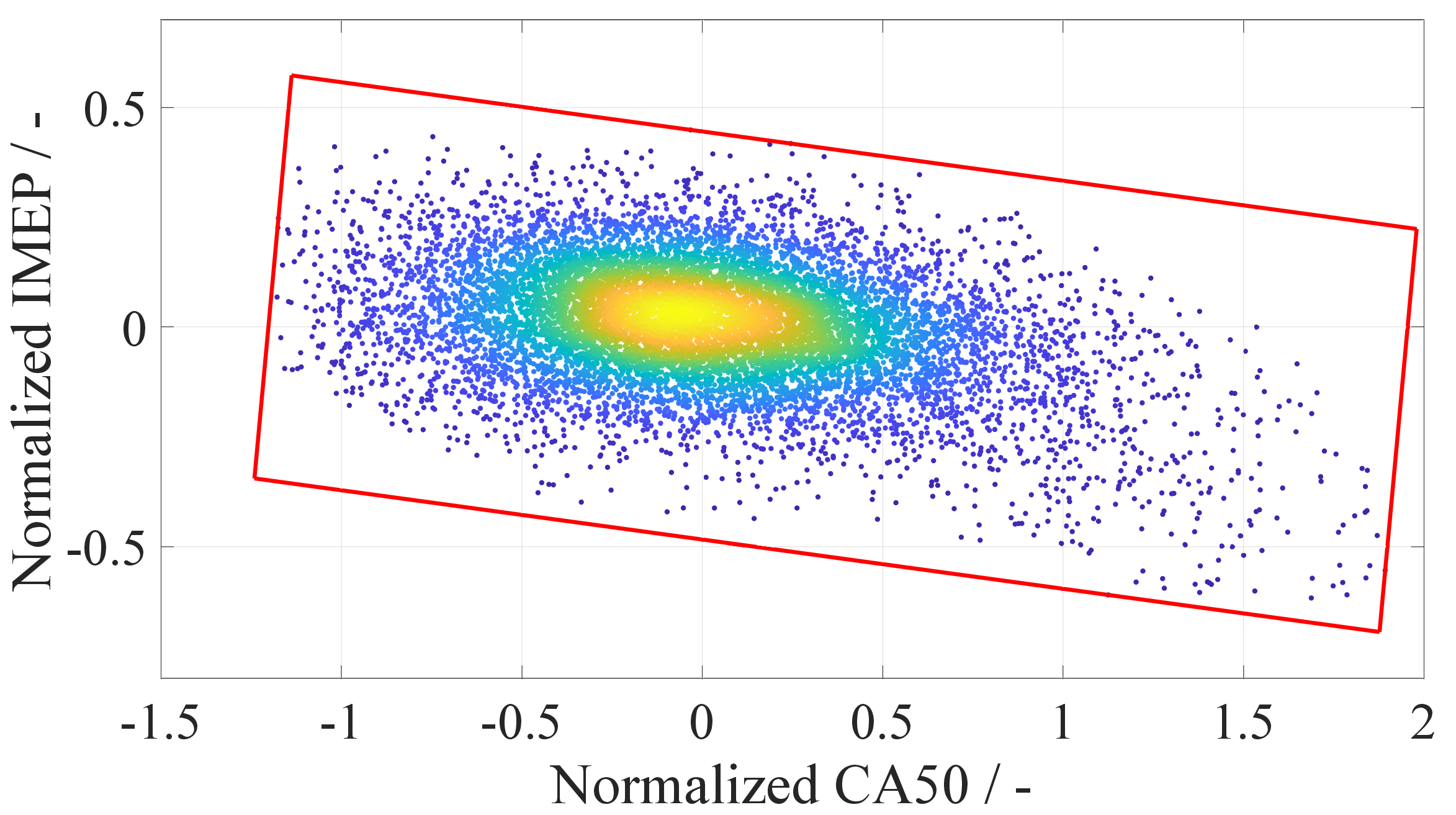}    
\caption{Historical data on the disturbances of CA50 and IMEP, and their identified disturbance margins.}  
\label{fig1}                                              
\end{figure}
In Figure~\ref{fig1}, we employed the methodologies outlined in \cite{cheramin2021data} to identify the boundaries of disturbance. These boundaries constitute the disturbance set in \eqref{W} required for our analysis. Utilizing a more refined partition could reduce conservatism to some extent, however, it would also result in increased computational complexity. In practical engine cases, the inherent disturbance evolves in response to changing environmental conditions. The conditional disturbances are typically distributed within the identified boundaries in Figure~\ref{fig1}. Therefore, we use the most recent disturbance data to estimate its conditional distribution.  \par
Then, we compare the performance of CA50 and IMEP under DA-DR MPC and tube-based MPC. The results below illustrate the controllers' impact on CA50 variation as well as the average performance of the system.
\begin{figure}[htb]
\centering
\includegraphics[width=0.475\textwidth]{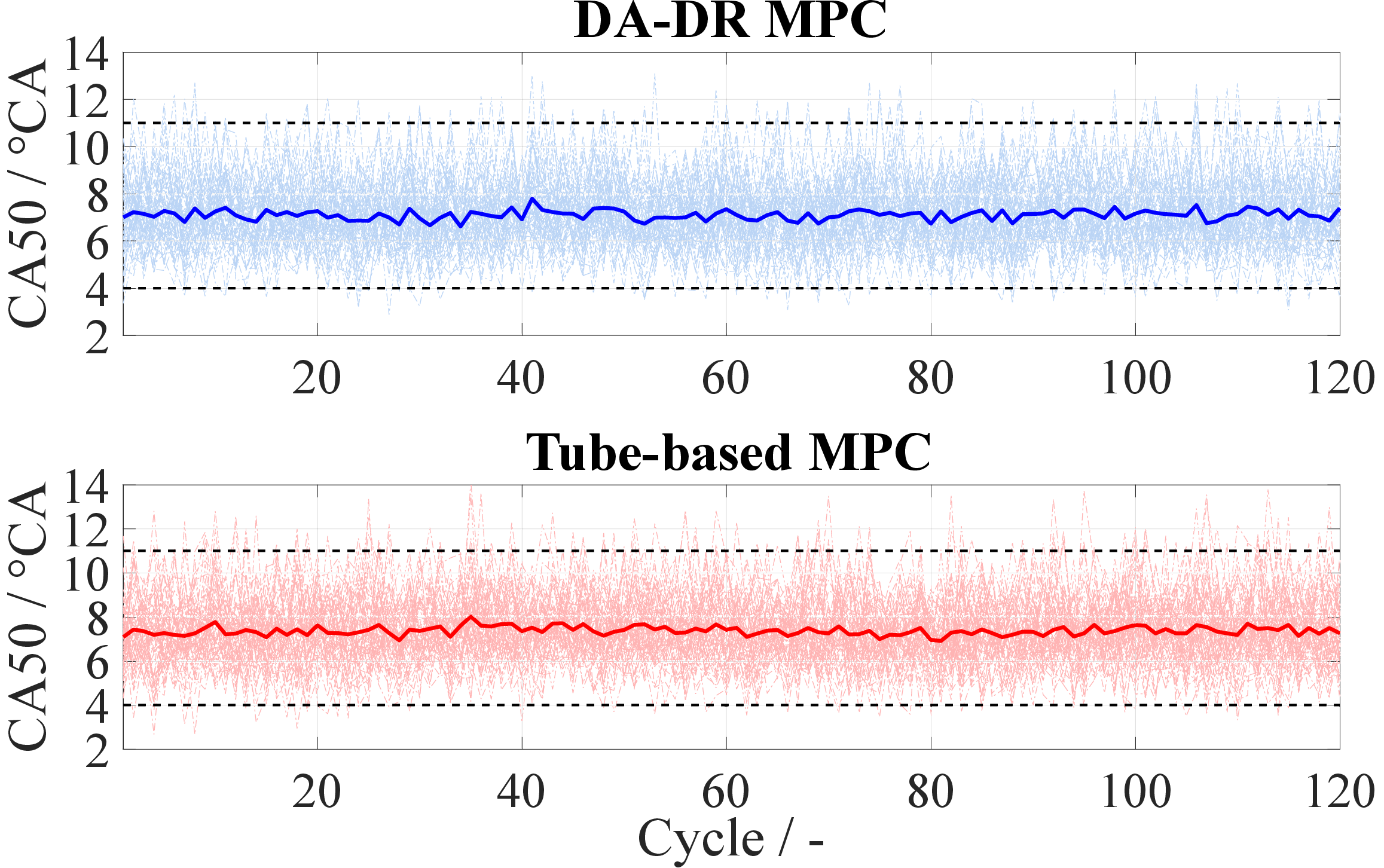}    
\caption{Variation comparison in CA50 under DA-DR MPC and tube-based MPC. The black dashed line represents the hard state constraint for CA50. Mean performance trajectories are depicted by the deep blue (DA-DR MPC) and deep red (tube-based MPC) lines. Light-colored lines illustrate individual simulation results under random disturbances, with a total of 60 simulations sampled.}  
\label{fig2}                                
\end{figure} \par
\begin{table}[htb]
\begin{center}
\caption{Statistical analysis of CA50 under two MPCs}
\begin{tabular}{|c|c|c|c|}
\hline
&Variance  & CA50$\ >11 $  & CA50$\ <4$ \\ \hline
DA-DR &2.2598& 1.52 $\%$ & 0.99 $\%$ \\ \hline
tube-based&2.4316 & 2.07 $\%$ & 0.77 $\%$ \\ \hline
\end{tabular}
\end{center}
\label{table1}
\end{table}\par
From Figure~\ref{fig2} and Table~\ref{table1}, we can observe that both controllers exhibit similar frequencies of violation with respect to the lower bound at $4$ °CA, with the tube-based MPC slightly outperforming the DA-DR MPC in this regard. However, when considering the variance and the frequency of upper bound violations across $60$ simulation runs, the DA-DR MPC demonstrates a smaller variance of $2.2598$ and a lower frequency of upper bound violations. Particularly in terms of average performance, DA-DR MPC can achieve superior performance compared to the tube-based method, as evidenced in Figure~\ref{fig3}.\par

\begin{figure}[htb]
\centering
\includegraphics[width=0.475\textwidth]{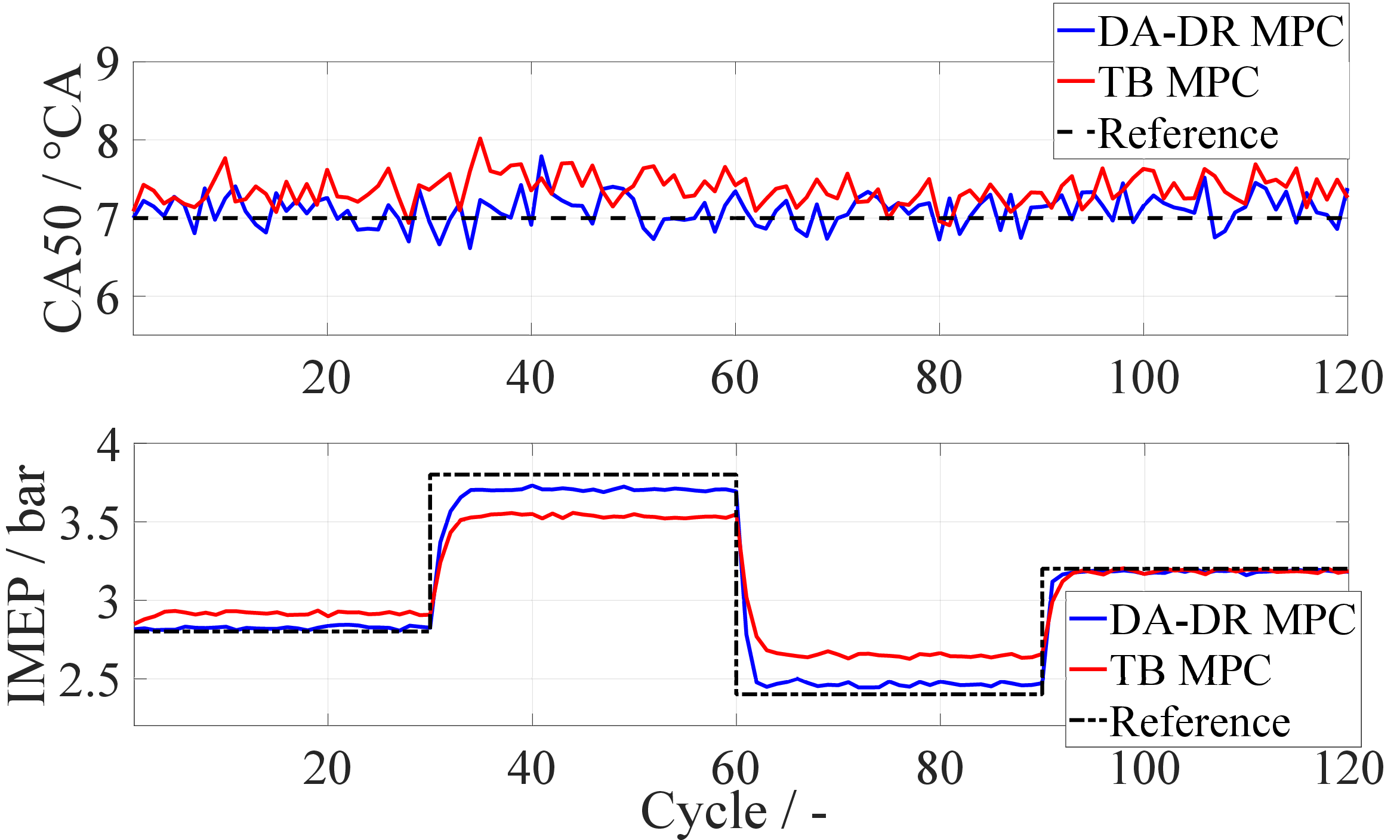}   
\caption{Average performance of CA50 and IMEP under DA-DR MPC and tube-based MPC control strategies. The blue lines and red lines represent DA-DR MPC and tube-based MPC, respectively. }  
\label{fig3}                                       
\end{figure}
Figure~\ref{fig3} illustrates the average performance of CA50 and IMEP under the two control strategies. Our goal is to maintain CA50 at the set point of $7$ °CA while ensuring that the controllers accurately track the IMEP reference. From Figure~\ref{fig3}, it is evident that under tube-based MPC control, the mean CA50 is higher than the set point, whereas the DA-DR MPC effectively maintains the mean CA50 around the desired set point. Furthermore, as shown in the lower portion of the figure, DA-DR MPC exhibits better tracking performance compared to the tube-based MPC, with the mean IMEP closer to the reference and presenting a smaller offset. In contrast, the tube-based MPC demonstrates suboptimal tracking performance at both high and low IMEP levels. Additionally, the DA-DR MPC responds more rapidly to abrupt reference changes. This is attributable to the DA-DR MPC's ability to provide greater control flexibility while satisfying the specified constraints. Conversely, the tube-based MPC adopts a more restricted action due to the conservative constant feedback approach to suppress the propagation of uncertainty. As a result, when chance constraints are all satisfied, the DA-DR MPC outperforms the tube-based MPC in terms of both average performance and variance, which are evaluated by the objective function. Besides, we conduct a Wilcoxon Rank Sum Test to compare two simulation groups, with the sample size of $60$. The test yields a p-value of $0.0016$, indicating a statistically significant difference in these two methods.

\section{Conclusion}
In this study, we proposed a distributionally robust stochastic MPC framework that operates without prior knowledge of specific disturbance distributions. By leveraging a disturbance-affine feedback strategy, our approach effectively manages uncertainty propagation and enlarges the feasible set of the Stochastic MPC problem compared to the tube-based approaches. Additionally, the recursive feasibility and the stability can be guaranteed within the proposed MPC framework. Simulation results demonstrate that our method outperforms tube-based MPC in both average performance and state variance control, showcasing its less control conservatism and strong potential for practical applications. 
However, the computational complexity of the DA-DR MPC remains a challenge, particularly in ensuring its recursive feasibility. Future work will focus on enhancing computational efficiency to enable real-time implementation in real-world systems.

\begin{ack}                               
This study was conducted as part of Research Unit Forschungsgruppe (FOR) 2401, and was supported by the German Research Foundation (DFG). We are deeply grateful for their valuable support. Additionally, we would like to extend our sincere gratitude to Professor Heike Vallery and Dr. Maximilian Basler for their valuable suggestions and feedback on this work.

\end{ack}
\bibliographystyle{plain}       
\bibliography{autosam}  

\appendix
\section{Appendix}
In this appendix, we present the previously defined matrices that were not clarified and some of the intermediate steps involved in the proofs.

\subsection{Matrices in Equation~\texorpdfstring{\eqref{xk}}{(xk)}}
The matrices $L_x \in \mathbb{R}^{n_xN_h \times n_x}  $, $L_u \in \mathbb{R}^{n_xN_h \times n_uN_h} $ and $L_w \in \mathbb{R}^{n_xN_h \times n_wN_h} $ are given by
\begin{equation} \label{matrix_Lx_Lu}
\begin{aligned}
   & L_x = \begin{bmatrix}
A \\
A^2 \\
\vdots \\
A^{N_h}
\end{bmatrix}, \quad
L_u = \begin{bmatrix}
B & 0 & 0 & \cdots & 0 \\
AB & B & 0 & \cdots & 0 \\
\vdots & \vdots & \vdots & \ddots & \vdots \\
A^{N_h-1}B & A^{N_h-2}B & \cdots & \cdots &  B
\end{bmatrix}, \\
&L_w = \begin{bmatrix}
D & 0 & 0 & \cdots & 0 \\
AD & D & 0 & \cdots & 0 \\
\vdots & \vdots & \vdots & \ddots & \vdots \\
A^{N_h-1}D & A^{N_h-2}D & \cdots & \cdots & D
\end{bmatrix}.
\end{aligned}
\end{equation}

\subsection{Proof of \texorpdfstring{Equation~\eqref{fb_con}}{Equation (fb\_con)}}
\begin{pf}
The intermediate reformulations in \eqref{fb_con} can be expressed as follows:
    \begin{equation*} 
\begin{aligned}
    &\max_{\mathbf{w}_{[0,t-1]} \in \mathcal{W}^{t}}\tilde{\mathbf{F}}_{tj} L_x x_k +\tilde{\mathbf{F}}_{tj}  L_u \mathbf{c}_k +  \tilde{\mathbf{F}}_{tj}  L_u \mathbf{K}_k \mathbf{w}_k\\
   & +  \tilde{\mathbf{F}}_{tj} L_w \mathbf{w}_k  -  \tilde{\mathbf{f}}_{tj} \leq 0 \\
\underset{ (a)}{\Leftrightarrow} &\max_{\mathbf{w}_{[0,t-1]} \in \mathcal{W}^{t}} \tilde{\mathbf{F}}_{xj}x_k+\tilde{\mathbf{F}}_{uj} \mathbf{c}_k + \mathbf{v}_k^{\mathrm {T}   }\tilde{\mathbf{V}}_j^{\mathrm {T}   }\mathbf{w}_k
+\tilde{\mathbf{F}}_{wj} \mathbf{w}_k\leq \tilde{\mathbf{f}}_{tj} \\
\underset{ (b)}{\Leftrightarrow} &\max_{\mathbf{w}_{[0,t-1]} \in \mathcal{W}^{t}} \tilde{\mathbf{F}}_{xj}x_k+\tilde{\mathbf{F}}_{uj} \mathbf{c}_k + \mathbf{v}_k^{\mathrm {T}   }\tilde{\mathbf{V}}_j^{(1)\mathrm {T}   }\mathbf{w}_{[0,t-1]}\\
&+ \mathbf{v}_k^{\mathrm {T}   }\tilde{\mathbf{V}}_j^{(2)\mathrm {T}   }\mathbf{w}_{[t,N-1]}+\tilde{\mathbf{F}}_{wj}^{(1)} \mathbf{w}_{[0,t-1]}+\tilde{\mathbf{F}}_{wj}^{(2)} \mathbf{w}_{[t,N-1]} \leq \tilde{\mathbf{f}}_{tj}\\
\end{aligned}
\end{equation*}
\begin{equation} 
\begin{aligned}
\underset{ (c)}{\Leftrightarrow} &\min_{y_{tj}}\tilde{\mathbf{F}}_{xj}x_k+\tilde{\mathbf{F}}_{uj} \mathbf{c}_k + \mathbf{v}_k^{\mathrm {T}   }\tilde{\mathbf{V}}_j^{(2)\mathrm {T}   }\mathbf{w}_{[t,N-1]}+\tilde{\mathbf{F}}_{wj}^{(2)} \mathbf{w}_{[t,N-1]}\\
&-\tilde{\mathbf{f}}_{tj}+\left ( I_{t\times 1}\otimes h \right ) ^{\mathrm{T} }y_{tj}\leq 0, \quad y_{tj} \geq  0,\\
&\left( I_{t\times t}\otimes h \right ) ^{\mathrm{T} }y_{tj}- \tilde{\mathbf{V}}_j^{(1)  }\mathbf{v}_k-\tilde{\mathbf{F}}_{wj}^{(1)\mathrm{T}} \geq 0,\\
\underset{ (d)}{\Leftrightarrow} &\tilde{\mathbf{F}}_{xj}x_k+\tilde{\mathbf{F}}_{uj} \mathbf{c}_k + \mathbf{v}_k^{\mathrm {T}   }\tilde{\mathbf{V}}_j^{(2)\mathrm {T}   }\mathbf{w}_{[t,N-1]}+\tilde{\mathbf{F}}_{wj}^{(2)} \mathbf{w}_{[t,N-1]}\\
&-\tilde{\mathbf{f}}_{tj}+\left ( I_{t\times 1}\otimes h \right ) ^{\mathrm{T} }y_{tj}\leq 0, \quad y_{tj} \geq  0,\\
&\left( I_{t\times t}\otimes h \right ) ^{\mathrm{T} }y_{tj}- \tilde{\mathbf{V}}_j^{(1)  }\mathbf{v}_k-\tilde{\mathbf{F}}_{wj}^{(1)\mathrm{T}} \geq 0,\\
\end{aligned}
\end{equation}
Step $(a)$ corresponds to the transformation of the decision variables into vector form, as delineated in \eqref{chance_con}. In step $(b)$, we truncate the relevant vector $\mathbf{w}_k$ and consider the worst-case regarding the first $t$ disturbance realizations. In step $(c)$, we introduce the dual variable $y_{tj}$ and employ the formulated dual problem. In step $(d)$, the minimization on the left-hand side of the inequality is interpreted as the existence of a feasible solution.
\end{pf}

\subsection{Proof of \texorpdfstring{Equation~\eqref{upper_b}}{Equation (upper\_b)}}
\begin{pf} Next, we focus on the conditional expectation of the difference between the optimal objective at $x_{k+1}$ and the optimal objective at $x_k$, given $x_k$. This can be reformulated as follows:
\begin{equation} 
\begin{aligned}
 &\mathbb{E} \left [ J^\ast(x_{k+1})- J^\ast(x_{k})| x_k\right]    \\
\leq &\mathbb{E}  [ \tilde{J}(x_{k+1})- J^\ast(x_{k})| x_k] \\
= & \mathbb{E}  [ \left \| (A+BK_f)x_{k+N_h}+w_{k+N_h} \right \|_P^2+ \left \| x_{k+N_h} \right \|_Q^2\\
 &  +\left \| K_fx_{k+N_h} \right \|_R^2-\left \| x_{k} \right \|_Q^2-\left \| u_{k} \right \|_R^2-\left \| x_{k+N_h} \right \|_P^2 | x_k]\\
=& \mathbb{E}  [  \left \| (A+BK_f)x_{k+N}\right \|_P^2+\left \| w_{k+N_h} \right \|_P^2+\left \| x_{k+N_h} \right \|_Q^2\\
 &  +\left \| K_fx_{k+N_h} \right \|_R^2-\left \| x_{k} \right \|_Q^2-\left \| u_{k} \right \|_R^2-\left \| x_{k+N_h} \right \|_P^2 | x_k]\\
=& \mathbb{E}  [ \left \| w_{k+N_h} \right \|_P^2 -\left \| x_{k} \right \|_Q^2-\left \| u_{k} \right \|_R^2| x_k]\\
= &  \mathbb{E}  [  \mathrm {tr}( w_{k+N_h}w_{k+N_h}^{\mathrm {T} } P ) ]-     \left \| x_{k} \right \|_Q^2-\left \| u_{k} \right \|_R^2\\
=&\mathrm {tr}(\Sigma_w P )-\left \| x_{k} \right \|_Q^2-\left \| u_{k} \right \|_R^2.
\end{aligned}
    \end{equation}  
The second equality arises from the independence of $x_{k+N_h}$ and $w_{k+N_h}$, while the third equality follows the definition from the definition of the Lyapunov equation. In the final equation, $\Sigma_w $ represents the sample covariance of $w$.
\end{pf}

\subsection{Proof of \texorpdfstring{Equation~\eqref{exp_val}}{Equation (exp\_val)}}
\begin{pf} 
The equality concerning the expectation in \eqref{exp_val} can be derived through the following steps:
    \begin{equation} 
    \begin{aligned}
 &\mathbb{E} \left [ J^\ast(x_{k+1})| x_0\right]    \\
=&\int_{\mathbb{R}^{n_x} } J^\ast(x_{k+1}) p (x_{k+1}|x_0)\mathrm{d} x_{k+1}\\
=&\int_{\mathbb{R}^{n_x} } \int_{\mathbb{R}^{n_x} }J^\ast(x_{k+1}) p (x_{k+1},x_{k}|x_0)\mathrm{d} x_{k}\mathrm{d} x_{k+1}\\
=&\int_{\mathbb{R}^{n_x} }  \int_{\mathbb{R}^{n_x} }J^\ast(x_{k+1}) p (x_{k+1}|x_{k})p(x_{k}|x_0)\mathrm{d} x_{k}\mathrm{d} x_{k+1}\\
=&\int_{\mathbb{R}^{n_x} } \mathbb{E} \left [ J^\ast(x_{k+1})| x_k\right]  p(x_{k}|x_0)\mathrm{d} x_{k}\\
=& \mathbb{E} \left [ \mathbb{E} \left [ J^\ast(x_{k+1})| x_k\right] | x_0 \right]  ,
\end{aligned}
\end{equation}
where $p$ denotes a probability density function, and the third equation indicates the Markov property of the system.   \hfill $\square$
\end{pf}

\end{document}